\documentclass[10pt, conference]{IEEEtran}
\IEEEoverridecommandlockouts

\usepackage{times}  
\usepackage{epsfig}
\usepackage[TABBOTCAP]{subfigure}
\usepackage{tabularx}
\usepackage{graphicx} 
\usepackage{color}
\usepackage{xspace}
\usepackage{verbatim}
\usepackage{colortbl}
\usepackage[ruled,linesnumbered]{algorithm2e}

\usepackage{fontenc}
\usepackage[T1]{fontenc}
\usepackage{amsfonts}
\usepackage{amsmath,amssymb}
\usepackage{graphicx}
\usepackage{epsfig}
\usepackage{color}
\usepackage{cite}
\usepackage{epstopdf}
\usepackage{subfigure}
\usepackage{eso-pic}
\usepackage{flushend}
\usepackage{amssymb}
\usepackage{amsmath}
\usepackage{amsthm}

\usepackage{multirow}
\usepackage{hyperref}
\usepackage{cleveref}
\usepackage{booktabs}
\usepackage{bbm}
\usepackage{url}
\newtheorem{prop}{Proposition}

\newif\ifcomment
\commenttrue

\ifcomment

\newcounter{AFNumberOfComments}
\stepcounter{AFNumberOfComments}
\newcommand{\afnote}[1]{\textcolor{blue}{\small AF-\arabic{AFNumberOfComments}\stepcounter{AFNumberOfComments}: #1}}

\newcounter{MGNumberOfComments}
\stepcounter{MGNumberOfComments}
\newcommand{\mgnote}[1]{\textcolor{magenta}{\small MG-\arabic{MGNumberOfComments}\stepcounter{MGNumberOfComments}: #1}}

\newcounter{DRNumberOfComments}
\stepcounter{DRNumberOfComments}
\newcommand{\drnote}[1]{\textcolor{red}{\small DR-\arabic{DRNumberOfComments}\stepcounter{DRNumberOfComments}: #1}}

\newcounter{JRNumberOfComments}
\stepcounter{JRNumberOfComments}
\newcommand{\jrnote}[1]{\textcolor{cyan}{\small JR-\arabic{JRNumberOfComments}\stepcounter{JRNumberOfComments}: #1}}

\else
\newcommand\afnote[1]{}
\newcommand\mgnote[1]{}
\newcommand\drnote[1]{}
\newcommand\jrnote[1]{}
\fi

\newcommand{\APPROX}{\textsc{approx}\xspace}
\newcommand{\CLASS}{\textsc{class}\xspace}

\begin{document}

\title{Accelerating Deep Learning Classification with Error-controlled Approximate-key Caching}


 \author{
 \IEEEauthorblockN{Alessandro Finamore, James Roberts, Massimo Gallo, Dario Rossi}
 \IEEEauthorblockA{HUAWEI Technologies, France}
 }


\pagestyle{plain} 

\maketitle

\begin{abstract}
While Deep Learning (DL) technologies are a promising tool to solve networking problems that map to classification tasks, their computational complexity is still too high with respect to real-time traffic measurements requirements. To reduce the DL inference cost, we propose a novel caching paradigm, that we named  \emph{approximate-key caching}, which returns approximate results for lookups of selected input based on cached DL inference results. While approximate cache hits alleviate DL inference workload and increase the system throughput, they however introduce an \emph{approximation error}.
As such, we couple approximate-key caching with an error-correction principled algorithm, that we named \emph{auto-refresh}.
We analytically model our caching system performance for classic LRU and ideal caches, we perform a trace-driven evaluation of the expected performance, and we compare the benefits of our proposed approach with the state-of-the-art similarity caching -- this testifies the practical interest of our proposal.
\end{abstract}

\section{Introduction}\label{sec:intro}

Edge Artificial Intelligence (AI) is a promising technology for distributing intelligence at all levels of the network and is instrumental to the advent of self-driving networks~\cite{feamster18anrw}. However, the high computational cost associated with AI models impedes their adoption in applications requiring fast analytics. To mitigate this issue, common options include simplifying the architecture of the model (e.g., techniques such as quantization and knowledge distillation reduce the model size at the cost of lower accuracy~\cite{mobilenet-corr17,bnn-NIPS2016,distillation-NIPS15}) and accelerating computation (e.g., in addition to using hardware accelerators~\cite{googleTPU,hwcomparison}, techniques such as early exit and LCNN bypass the computation of certain model layers~\cite{lcnn-CVPR17,laskaridis20mobicom}).


Generally speaking, the key to alleviate DL computations is to reduce resource wastage due to \emph{repeated operations}. In other words, since DL models operate on inputs that are the same as a previous input (e.g., due to skew in the input popularity), it may be beneficial to avoid repeating computations by \emph{caching} the DL model output corresponding to recurring inputs. For example, Twitter observed that ``\emph{an increasing number of cache clusters are devoted to caching computation related data, such as features, intermediate and
final results of Machine Learning (ML) prediction [which accounts for] 50\% of all Twemcache clusters}''~\cite{yang2020osdi}, while Microsoft stated that ``\emph{machine learning algorithms
occupy hundreds of machines for tens of milliseconds to
select the ads for each query}''~\cite{li-WWW18} -- caching is key for DL-based real-time system performance. 
 
While \emph{exact caching} is already popular for DL inference systems~\cite{clipper-NSDI17,laser-WSDM14,velox-CIDR15}, we argue that a more flexible paradigm is required: for similar inputs, it may be advantageous to use cached DL model results as \emph{approximate DL results}. This further reduces the DL inference workload at the cost of an approximation error. Providing cached results for approximately similar queries, i.e., \textit{similarity caching}, is commonplace in content retrieval systems. However, due to the nature of the queries, previous literature in this area~\cite{virage-96,luo2002sigmod,falchi-LSDS08,chierichetti09pods,isax-icdm10,imcom17,lucene-10,pandey-WWW09,li-WWW18,garetto20infocom} is not necessarily suitable for ML/DL \emph{classification} use cases, which are common in the networking domain. 


In this work, we introduce a novel approximate caching technique, called \emph{approximate-key caching}, that improves on \emph{similarity caching}. As for similarity caching, approximate-key caching answers to queries with approximate results. However, approximate-key caching differs from similarity caching in two main aspects: ($i$) it lets the user design an approximation function that is well-adapted to the use case at hand, bringing improved control over the realized hit rate, and ($ii$) it introduces a verification of the key-value pairs stored in the cache using a mechanism we call \emph{auto-refresh} that \emph{explicitly} controls approximation errors. Our contributions are as follows, 
\begin{itemize} 
    \item We introduce \emph{approximate-key} caching, a new caching paradigm retaining the simplicity of an \emph{exact matching} lookup while significantly \emph{increasing the hit rate}.
    \item We design an \emph{auto-refresh} mechanism to automatically trade-off between exploitation (to reduce DL inference workload) and exploration (to verify cached results for error-control). 
    \item We analytically model the error of approximate-key caching, with and without error correction, for LRU replacement (numerical) and an ideal cache (closed-form).
    \item We present a thorough trace-driven evaluation of approximate-key caching, including a comparison with state-of-the-art similarity caching.
\end{itemize}

In the remainder of this paper, we first provide the necessary background on classification and caching (Sec. \ref{sec:background}).
We then introduce approximate-key caching, detailing the auto-refresh error-control mechanism (Sec. \ref{sec:approxkey}) for which we provide an analytical model (Sec. \ref{sec:model}). We next perform a trace-driven evaluation using a real dataset from a traffic classification use case, demonstrating the soundness of the proposed mechanism, and illustrating the benefits it offers compared to both exact caching and similarity caching (Sec. \ref{sec:evaluation}). Finally, we place our work in the context of related literature  (Sec. \ref{sec:related}) and summarize our findings (Sec. \ref{sec:conclusion}).

\section{Classification and caching}\label{sec:background}

We start by formalizing the notion of classification tasks and discussing the use of caching to reduce the compute load. Then, we introduce \textit{exact caching} and \textit{similarity caching}, identify their limitations for classification tasks, and  motivate the need for a novel alternative mechanism.


\subsection{Classification tasks}

We focus on scenarios where an input from an infinite stream needs to be mapped to a class in real-time.
The set of possible classes $\mathcal{C}$ is relatively small (e.g., a few hundred), while the input is a vector $x$ (of discrete or continuous-valued elements) from a very large space $\mathcal{X}$. 
Such a classification task realizes the deterministic mapping $y= \CLASS(x)$ from an input $x\in \mathcal{X}$ to a class $y \in \mathcal{C}$ as depicted in Fig.~\ref{fig:exact-approx-synoptic}-(a). 

\begin{figure}[!t]
    \centering
    \includegraphics[width=\columnwidth]{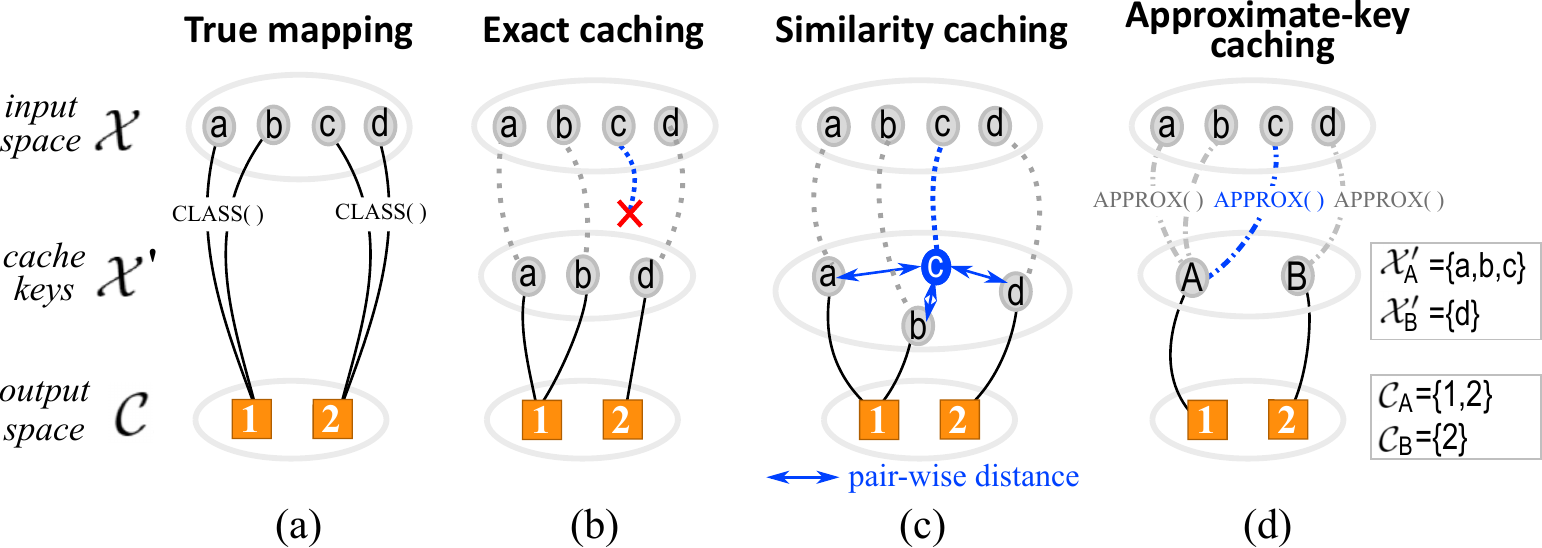}
    \caption{Synoptic of Exact, Similarity, and Approximate-key caching schemes.}
    \label{fig:exact-approx-synoptic}
\end{figure}

Typical $\CLASS(\cdot)$ functions are both memory hungry and computationally intensive. For instance, DL traffic classification models \cite{aceto-tma18} require 300k-6M weights, while the inference duration of smaller state-of-the-art architectures optimized for real-time is of the order of 150 to 250ms~\cite{mobilenet-corr17}.
Thus, for DL-based network traffic processing, reducing the computational complexity is key to meeting peak demand. Rather than reducing complexity by simplifying the $\CLASS(\cdot)$ function (sacrificing classification accuracy), in this paper we explore the orthogonal and general approach of avoiding redundant computations via caching. 




\subsection{Exact caching}
\label{sec:exact}
If the input space $\mathcal{X}$ is discrete, one can employ \emph{exact caching} to store key-value pairs $(x,y)$ while retaining the deterministic nature of the classification. If an input key $x$ is stored in the cache, its value $y$ can be returned immediately. As long as the cache hit rate is sufficiently high ~\cite{clipper-NSDI17,laser-WSDM14}, we can significantly reduce the average response time 
(e.g., sub-microsecond for DRAM access vs hundreds of milliseconds for DL inference) and save computational resources for processing less popular input keys.
We exemplify this in Fig.~\ref{fig:exact-approx-synoptic}-(b) where the cached $\{(a,1), (b,1), (d,2)\}$ pairs avoid invoking $\CLASS(\cdot)$ for $\{a,b,d\}$. 

Exact caching is simple to implement, and the widely used \emph{Least Recently Used} (LRU) replacement policy is known to offer good performance in practice: if an input $x$ is not in the cache, its value $y= \CLASS(x)$ is retrieved and the key-value pair $(x,y)$ is cached; to make room, the least recently accessed cached pair is evicted. 
We are aware that the cache literature is ripe with alternative policies that can achieve a more favorable tradeoff between hit rate, algorithm complexity, and reactivity to changes in popularity~\cite{martina14infocom}. Rather than exploring this space, we consider instead a hypothetical \emph{ideal caching} policy that, in a cache of capacity $K$, caches precisely the $K$ most popular key-value pairs. As we shall see, this facilitates the study of alternative methods to exact caching. However, even ideal exact caching is ineffective for many classification tasks. In fact, the input space $\mathcal{X}$ is commonly too large, and for any reasonable cache capacity $K$, the fraction of inputs that match a cache entry is too small to offer a relevant hit rate (Sec. \ref{sec:evaluation}).

%

\subsection{Similarity caching}
\label{sec:similarity}
To circumvent the issues of exact caching, many authors have proposed to answer cache lookups with approximate results,  trading off precision for a higher hit rate. Among a number of application-specific approaches, the more generic notion of \emph{similarity caching} has emerged (also known as ``metric caching''~\cite{falchi-LSDS08}, and ``nearest neighbor caching''~\cite{pandey-WWW09}), notably for applications that use \emph{approximate similarity search}. 

Approximate  similarity  search~\cite{simsearch-survey01,simsearch-survey05,simsearch-survey06} builds on the notion that if two points $x_1$ and $x_2$ are sufficiently close to each other, 
$y_1$ and $y_2$ will also be close. 
The principle is to respond to an input query $x_1$ with a cached value $y_2$ where the key $x_2$ associated with $y_2$ is similar but not necessarily identical to $x_1$, i.e., when a metric $\textsc{dist}(x_1, x_2)$ is less than a given \emph{distance threshold} $\epsilon$. The appropriate definition of  $\textsc{dist}(x_1, x_2)$ and the choice of threshold $\epsilon$ clearly depend on the particular use case, and both contribute to the tradeoff between cache hit rate 
and response accuracy. This commonly results in implementing a k Nearest Neighbor (kNN) strategy.
For instance, in Fig.~\ref{fig:exact-approx-synoptic}-(c) the closest cached input to $c$ is $b$, which generates an error by returning $1$ rather than~$2$; alternatively, by setting a very small $\epsilon$, a lookup for $b$ would generate a cache miss.

While similarity caching is an appealing concept, it presents downsides in our settings.
First, similarity caching assumes that the closeness of two inputs also implies that their respective output values are close -- otherwise stated, it assumes that objects in the neighborhood of a cached input offer a \emph{low approximation error} for a variety of other inputs. As such, similarity caching is particularly useful for searching, 
i.e., when the value associated with a key is not unique but rather a set of possibilities. However, mapping this scenario to a classification task requires ingenuity, as it is not obvious how to define an appropriate similarity distance, or how to tune the threshold $\epsilon$ to balance hit rate against accuracy.

Second, the lookup function performed in similarity caching is much more complex than an exact-match lookup:
finding one or more nearest neighbors to an input $x$  may be particularly time-consuming when the input space $\mathcal{X}$ is large, even with state-of-the-art algorithms such as ball trees, k-d trees, and Locality Sensitive Hashing (LSH)~\cite{knn-book}.

Third, it is necessary to define a replacement policy for similarity caching that ensures that the cache is populated with key-value pairs adequately covering the input space. The selection of which cached pair is least useful and must be evicted is significantly more complex (as it triggers an update of the ball tree, LSH, or other data structure used) than the LRU policy for exact caching~\cite{falchi-LSDS08,garetto20infocom}.

Last but not least, whereas similarity caching introduces classification errors, it does not provide any direct means to correct them. In fact, while the $\epsilon$ distance threshold controls the similarity between different inputs, there is typically no formal verification of the cached $(x,y)$ pairs. In other words, errors are only ``indirectly'' controlled by LRU-like eviction mechanisms.

To address the limitations of exact caching and similarity caching we have devised an alternative form of approximate caching that we term \emph{approximate-key caching}.

\section{Approximate-key caching with error control}
\label{sec:approxkey}

We first introduce approximate-key caching (Sec. \ref{sec:AKC}) before explaining how the error rate can be controlled with our proposed \emph{auto-refresh} algorithm (Sec. \ref{sec:autorefresh_algo}). We contrast approximate-key caching with exact and similarity caching in Table \ref{tab:design-compare} and in the toy example in Fig.~\ref{fig:exact-approx-synoptic}-(d).

\subsection{Approximate-key caching}
\label{sec:AKC}
\begin{table}[t]
\centering
\caption{Taxonomy of the cache design space}
\label{tab:design-compare}
\begin{tabular}{cccccc}
\toprule
\bf Caching &
\bf Key &
\bf Match &
\bf Hit rate &
\bf Cost &
\bf Err control
\\
\midrule
exact               & original     & exact      &low &low & not needed\\
similarity          & original     & approx   &high &high & indirect\\
approx-key  & approx  & exact   &high &low & direct\\
\bottomrule
\end{tabular}
\end{table}

Approximate-key caching transforms the input space ($i$) to significantly reduce the size of the original input space $\mathcal{X}$ that makes exact caching impractical due to a very low hit rate, and ($ii$) to enable the use of exact matching in the transformed space. The transformation should ideally increase the key popularity skew to improve the hit rate. More formally, instead of caching inputs $x$$\in$$\mathcal{X}$ (as happens in similarity caching), we cache approximate inputs $x'=\APPROX(x) \in \mathcal{X'}$ where $|\mathcal{X'}| \ll |\mathcal{X}|$.  

Fig.~\ref{fig:method_sketch}-(right) illustrates some examples of $\APPROX(\cdot)$ functions on a time series input $x$ having six integer elements:
\begin{itemize}
\item \emph{prefix$_3$} (\emph{suffix$_3$})  takes the first (last) 3 elements;
\item  \emph{every$_2$} (\emph{maxpool$_2$}) takes every second element (the max between two consecutive elements);
\item  \emph{quantize$_{10}$}, rounds elements to the nearest multiple of 10.  
\end{itemize}
Clearly, other $\APPROX(x)$ functions can be defined, including combinations of some of the above. The choice depends on the precise nature of the considered use case. These transformations make caching more flexible with respect to similarity caching where the cache keys are kept unmodified.

When an approximate key $x'$ is cached, it remains to define the associated value $y'$. One possible approach is to store the value inferred by the $\CLASS(\cdot)$ function whenever an approximate key is inserted into the cache: if $x'=\APPROX(x)$ is not in the cache, then a classification inference is triggered and the resulting key-value pair $(x',y')$ is cached, where $x'=\APPROX(x)$ and $y'=\CLASS(x)$. 
Subsequently, as long as $x'$ remains in the cache, every input key that is approximated by $x'$ will find the stored value $y'$ with a \emph{simple exact match} operation. 
The replacement policy might be any of those applicable to exact caching, including LRU and the hypothetical ideal policy discussed in Sec.~\ref{sec:exact}.

This process is not immune to errors since not all $x$ mapping to the same $x'$ share the same class. For instance, in Fig.~\ref{fig:exact-approx-synoptic}-(d) the approximate-key $A$ maps to the inputs $\mathcal{X}_A=\{a,b,c\}$ which have different classes $\mathcal{C}_A=\{1,2\}$; the approximate-key $B$ instead does not generate mismatches.

Approximate-key caching, as defined above, may be satisfactory if popular approximate keys (like $B$ in the illustration) correctly represent the class of all, or a very large fraction of the inputs mapping to them -- i.e., the popular approximate keys present a \emph{dominant class}. However, this can hardly be guaranteed, so we introduce a \emph{direct} verification of the $(x',y')$ mappings. The purpose of such a mechanism is to ensure that the cache preferentially stores key-value pairs that yield no error (like $B,$ class 1), while cases with mismatches (as $\mathcal{X}'_A$) require multiple verifications with $\CLASS(\cdot)$.

\begin{figure}[!t]
\centering
\includegraphics[width=\columnwidth]{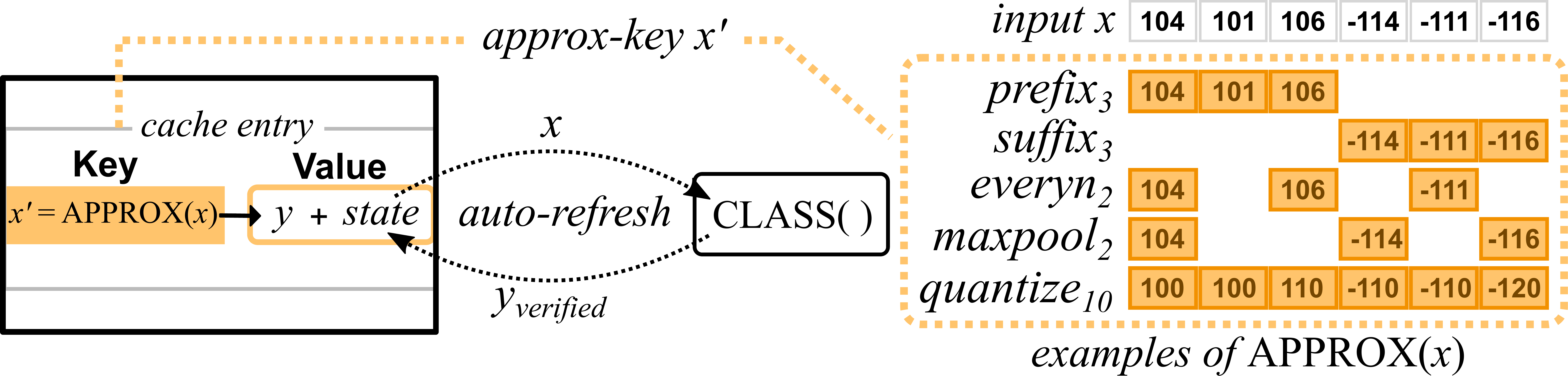}
\caption{Approximate-key caching system overview. 
\label{fig:method_sketch}
}
\end{figure}

\subsection{Error control via auto-refresh}\label{sec:autorefresh_algo}

\newcommand{\algocomment}[1]{\textcolor{blue}{#1}}
\newcommand{\algrule}[1][.2pt]{\par\vskip.5\baselineskip\hrule height #1\par\vskip.5\baselineskip}


\begin{algorithm}[t]
\footnotesize

\SetAlFnt{\small}
\SetArgSty{textnormal}
\SetAlgoNoLine

\KwIn{cache object, input $x$, back-off $\beta$;}
\KwOut{$y$ class;}
\algrule

$x' = \APPROX(x)$                             
\algocomment{\hspace{50pt} \# compute approximate-key (fast)}
\\
($y$, $state$) = cache.lookup($x'$)      
\algocomment{\hspace{9pt} \# exact matching  (fast)}
\\
\uIf
    (\algocomment{\hspace{46.8pt} \# miss: add new entry})
    {$y$ is \texttt{null}}
    {                           
        $y = \CLASS(x)$                   
        \algocomment{\hspace{43pt} \# inference  (slow)} 
        \\
        $state.to\_serve$ = 0
        \\
        $state.refreshed$ = 1 
        \\
        cache.add($x'$, ($y$, $state$))                
    }
\uElseIf
    (\algocomment{\hspace{2.5pt} \# hit (no refresh)})
    {$state.to\_serve$ > 0}
    {
        $state.to\_serve$ -= 1
    }
\uElse
    (\algocomment{\hspace{95pt} \# hit (refresh needed)})
    {
        $y_{verify}$ = \CLASS(x)
        \algocomment{\hspace{21pt} \# inference (slow)}
        \\
        \uIf
            (\algocomment{\hspace{18pt} \# no conflict: increase back-off})
            { $y_{verify}$ == $y$} 
            {
                $state.to\_serve$ = floor(pow($\beta$, $state.refreshed$)) \\
                $state.refreshed$ += 1 \\
            }
        \uElse
            (\algocomment{\hspace{79pt} \# conflict: reset ($y$, state)}) 
            {
                $y$ = $y_{verify}$ \\
                $state.to\_serve$ = 0 \\
                $state.refreshed$ = 1 \\
                cache.update($x'$, ($y$, $state$))
            }
    }
\textbf{return} y
\caption{Auto-refresh error correction}
\label{alg:approx-autorefresh}
\end{algorithm}

We named our error control algorithm \emph{auto-refresh} and Fig.~\ref{fig:method_sketch}-(left) sketches our full-fledged approximate-key caching system. Each cache value is composed of the class $y$ and a \emph{state} used to track the validity of the mapping ($x'$, $y'$). The mechanism verifies the accuracy of the mapping for selected inputs and updates the value if there is a mismatch. If the verification is successful, the interval before the next verification is increased exponentially. The pseudocode of the algorithm is presented in Algorithm \ref{alg:approx-autorefresh}.

When an input $x$ is processed, the algorithm first computes approximate-key $x'$ and does a cache lookup (with exact matching) to obtaining the related value ($y$, $state$). If there is a miss, inference $y=\CLASS(x)$ is run and the pair $(x',(y, state))$ is added to the cache with a reset \emph{state} (lines 3-7). 
The counter $state.to\_serve$ defines how many lookups can be served before the next refresh, while $state.refreshed$ counts the number of refreshes triggered so far.

If $x'$ is in the cache (hit), two cases are possible. When no refresh is needed, the state is simply updated (line 8-9). Otherwise, a verification takes place invoking $\CLASS(\cdot)$: if the current and verified classes are consistent, the number of lookups that can be served without refresh is exponentially increased (with base $\beta > 1$); if the classes mismatch, $y_{verify}$ replaces $y$ and the $state$ is reset (lines 11-19).

Otherwise stated, the algorithm verifies every input after a new cache update 
until the number of inputs $n$, including the first, exceeds $\beta^{n-1}$. This means that the verification is initially frequent, especially if $\beta$ is small, and then becomes exponentially rarer as long as repeated matches confirm the stored class 
(i.e., the entry has a dominant class). This intuitively realizes the objective stated at the end of the previous subsection: the cache avoids costly inference for approximate keys with a low probability of error (like $B$ in Fig. \ref{fig:exact-approx-synoptic}-(d)), while approximate keys with a higher probability of error (like $A$ in Fig. \ref{fig:exact-approx-synoptic}-(d)) are subject to frequent verification. In the next section, we formally analyze these benefits.

\section{Analytical model}\label{sec:model}

We first analyze the performance of approximate-key caching without error control (Sec. \ref{sec:model:error}), and then extend the model to account for the auto-refresh mechanism (Sec. \ref{sec:model:autorefresh}). 

\subsection{Approximate-key caching without error control} 
\label{sec:model:error}
 
To formalize the discussion on errors mentioned above we first introduce some notation. It is convenient to refer to  $\mathcal{X'}$ members as $x'_i$, for $1 \le i \le |\mathcal{X'}|$. The relative $x'_i$ popularity is denoted $q_i$, i.e., the probability that an arbitrary input 
has approximate-key $x'_i$. We have $\sum_i q_i = 1$ and, without loss of generality, we order the $x'_i$ by their decreasing popularity, i.e., $q_i \le q_j $ if $i<j$. We shall assume the input stream follows the Independent Reference Model (IRM), a common assumption in the cache performance analysis literature, where the probability $q_i$ 
is independent of the order of $x$ in the input stream.

For \textit{LRU caching}, the characteristic time approximation yields the cache hit rate $H^{(\textsc{lru})}$ of a cache of capacity $K$ \cite{Fricker2012}. We hence have that,
\begin{equation}
H^{(\textsc{lru})} = \sum_{1\le i \le |\mathcal{X}'|} q_i (1-e^{-q_i t_c}),
\end{equation}
where the characteristic time $t_c$ solves the equation,
\begin{equation}
\sum_{i} (1-e^{-q_i t_c}) = K.
\end{equation}
The hit rate for inputs with approximate-key $x'_i$ is then $h_i = 1-e^{-q_i t_c}$.

An \textit{ideal cache} permanently stores the $K$ most popular approximate keys and therefore realizes a hit rate $h_i=1$, for $1\le i \le K$, and $h_i=0$, otherwise. The overall hit rate is, 
\begin{equation}
H^{(\text{ideal})} = \sum_{1\le i \le K} q_i. 
\label{eq:Hideal}
\end{equation}
 
For the subset of inputs with approximate-key $x'_i$, the true class belongs to a set $\mathcal{C}_i$, with members denoted $y_{ij}$ for $1 \le j \le m_i$ with $m_i=|\mathcal{C}_i|$. In extension of the IRM, we assume the class of an arbitrary input in this subset is $y_{ij}$ with probability $p_{ij}$, independently of its position in the input stream, with $\sum_j p_{ij}=1$.

Without error control, the class associated with a cached approximate-key $x'_i$ is that of the input that led to the cache insertion. Under the above independence assumptions, the probability an incorrect class is returned for this key is,
\begin{equation}
e_i = \sum_j p_{ij}(1-p_{ij}) = 1 - \sum_j p_{ij}^2.
\label{eq:nocontrol}
\end{equation}
The overall LRU and ideal caching error rate without correction are thus,
\begin{equation}
E_{nc}^{\text{(LRU)}} = \sum_i q_i h_i e_i, \quad E_{nc}^{\text{(ideal)}} = \sum_{i\le K}  q_i e_i.
\end{equation}
It is clear from \eqref{eq:nocontrol} that the error rate $e_i$ will be small if $\max_j \{p_{ij}\}$ is close to 1 (i.e., where there is a dominant class $y_{ij}$ mapping to the majority of inputs $x$ having approximate-key $x'_i$)
and rather high when the true class can take several values with similar, small probabilities (e.g., when all labels are equally likely $p_{ij}=1/m_i$ for all $j$, then $e_i = 1-1/m_i$).
The auto-refresh algorithm is designed to preferentially rely on the cached value of dominant labels while  resorting to an actual inference $\CLASS(\cdot)$ when  $\max_j \{p_{ij}\}$ is small.

\subsection{Error control via auto-refresh}
\label{sec:model:autorefresh}

We continue the analysis focusing on the effects of the auto-refresh algorithm. We first consider an LRU cache of capacity $K$ and derive the fraction $r_i$ of inputs with approximate-key $x'_i$ that require DL inference (due to insertion or refresh), and the corresponding fraction $e_i$ of errors due to class mismatch. We then simplify the model by assuming the cache is ideal, leading to closed-form results of more intuitive interpretation. 

\subsubsection{LRU replacement}\label{sec:model:autorefresh:lru}
We consider sequences of input arrivals with approximate keys $x'_i$ that begin with a new cache insertion and terminate with the last arrival before either ($i$) an LRU eviction, or ($ii$) an auto-refresh update due to a class mismatch. We refer to such a sequence 
as a $j$-sequence.

It is convenient at this point to formulate the result of the Algorithm \ref{alg:approx-autorefresh} as follows. 
If there is no prior mismatch, the $n^{th}$ inference of the $j$-sequence occurs on input $\phi_n$, counting the initial identification on input 1 as the first inference ($\phi_1=1$), where
\begin{equation}
\phi_n =  \max\{n, \lfloor \beta^{n-1} \rfloor \}.
\label{eq:en}
\end{equation}
For instance, if $\beta=2$, inferences occur on inputs $2^{n-1}$ for $n\ge 1$ while if $\beta=1.5$, inferences occur on inputs  1, 2, 3, 5, 7, 11,$\cdots$. 


Let $P_j^{\mathrm{mm}}(a)$ be the probability a $j$-sequence is of length $a$, $a\ge 1$, and ends due to a mismatch. The sequence ends just before the $n^{th}$  $\CLASS(\cdot)$ inference where $\phi_n=a+1$ and $n\ge 2$. It counts $a+1$ hits (including the mismatch), $n-2$ matching  $\CLASS(\cdot)$ inferences (after the first), and one mismatch (on arrival $\phi_n$). By the independence assumptions and applying the characteristic time approximation, we deduce,
$$ P_j^{\mathrm{mm}}(a) = \left\{
\begin{array}{ll}
        (1-e^{-q_it_c})^{a+1}p_{ij}^{n-2} (1-p_{ij}), & \text{if } a = \phi_n - 1,\\
       0, & \text{otherwise}.
\end{array} \right.
$$
Let $P_j^{\mathrm{lru}}(a)$ be the probability of a $j$-sequence being of length $a$ and ending due to a cache eviction prior to arrival $a+1$. The sequence then has $a$ cache hits, and is followed by one cache miss. It contains $n-1$ matching  $\CLASS(\cdot)$ inferences with $\phi_n \le a < \phi_{n+1}$. By independence we have therefore,
$$ P_j^{\mathrm{lru}}(a) = (1-e^{-q_it_c})^{a} e^{-q_it_c} p_{ij}^{n-1}, $$
with $n$ such that $\phi_n\le a < \phi_{n+1}$.

Now consider the probabilities $\pi_j$ that an arbitrary sequence is a $j$-sequence for $1 \le j \le m_i$. By stationarity, the $\pi_j$ satisfy recurrence relations,
$$ \pi_j = \sum_{k\ne j}\pi_k \sum_{x\ge 1} P_k^{\mathrm{mm}}(x) p_{ij}/(1-p_{ik})  + \sum_{k\ge 1}\pi_k \sum_{x\ge 1}P_k^{\mathrm{lru}}(x) p_{ij},$$
that, with the normalization condition $\sum_{j \ge 1} \pi_j = 1$, can be solved numerically.

Overall, the probability an arbitrary sequence is of length $a$~is,
$$ P(a) = \sum_j \left(P_j^{\mathrm{mm}}(a) + P_j^{\mathrm{lru}}(a)\right) \pi_j.$$ 
From this distribution we can derive the required performance measure $r_i$ as follows. The number of $\CLASS(\cdot)$ inferences including the first in a sequence of length $a$ is $n(a)=n: \phi_n \le a <\phi_{n+1}$. The overall fraction of arrivals that requires inference is thus,
\begin{equation}
r_i = \frac{\sum_{a\ge 1} n(a) P(a) }{ \sum_{a\ge 1} a P(a)}.
\end{equation}

To compute the fraction of errors $e_i$, consider a $j$-sequence of length $a$.  The number of unverified arrivals is $a-n(a)$ (since there are $n(a)$ $\CLASS(\cdot)$ inferences). Each such arrival independently gives an error with probability $1-p_{ij}$ so that the expected number of errors in the sequence of $a$ arrivals is $(1-p_{ij}) (a-n(a))$. We deduce the overall expected fraction of errors,
\begin{equation} e_i = \frac{\sum_a \sum_j (1-p_{ij})(a-n(a)) \pi_j\left(P_j^{\mathrm{mm}}(a) + P_j^{\mathrm{lru}}(a)\right)}{  \sum_{a\ge 1} a P(a)}.
\end{equation}

The above formulas can be used to evaluate the auto-refresh algorithm for a traffic model specified by the distributions $\{q_i\}$ and $\{p_{ij}\}$. However, the evaluation is necessarily numerical and the formulas provide little insight into the impact on performance of the traffic model and the refresh parameter $\beta$. To gain further understanding, we now consider the simpler model of an ideal cache.

\subsubsection{Ideal cache}\label{sec:model:autorefresh:ideal}
An ideal cache contains just the $K$ most popular items. This can be realized approximately in practice using more sophisticated replacement policies than simple LRU, as discussed in \cite{martina14infocom}. However, we consider it here more as a useful abstraction that allows us to evaluate the impact of auto-refresh on the most popular items that are indeed always in cache with very high probability. In the ideal cache, the sequences defined above always end because of a mismatch and the formulas simplify. The $\CLASS(\cdot)$  inference fraction and the error fraction for approximate-key $x'_i$, with $i\le K$, are given by the following proposition. 

\begin{prop}
\label{prop:analysis}
The fraction $r_i$ of inputs with approximate-key $x'_i$ that are verified by a $\CLASS(\cdot$) inference is given by
\begin{equation}
r_i= \left\{
\begin{array}{ll}
       \frac{1}{ \sum_j\sum_{n\ge 2}(\phi_n - 1) (1-p_{ij})^2 p_{ij}^{n-1}}, & \text{if } \max_j\{p_{ij}\}<1/\beta ,\\
       0, & \text{otherwise}.
\end{array} \right.
\label{eq:lookups}
\end{equation}
The probability $e_i$ an input with approximate-key $x'_i$ is incorrectly classified is
\begin{equation}
e_i =  \left\{
\begin{array}{ll}
        \frac{ \sum_j\sum_{n\ge 2}(\phi_n-n) (1-p_{ij})^3 p_{ij}^{n-1}} { \sum_j\sum_{n\ge 2}(\phi_n - 1) (1-p_{ij})^2 p_{ij}^{n-1}}, & \text{if } \max_j\{p_{ij}\} < 1/\beta,\\
       1- \max_j\{p_{ij}\}, & \text{otherwise}.
\label{eq:errors}
\end{array} \right.
\end{equation}

\end{prop}
\begin{proof}
Let $P_j(a)$ be the probability a sequence is of length $a$, $a\ge 1$, given that it begins with fresh class $y_{ij}$. By the independence assumption, we have for $n\ge 2$,
$$ P_j(a) = \left\{
\begin{array}{ll}
        p_{ij}^{n-2} (1-p_{ij}), & \text{if } a = \phi_n - 1,\\
       0, & \text{otherwise}.
\end{array} \right.
$$
The recurrence relations for the probabilities an arbitrary sequence begins with the insertion of $y_{ij}$ become,
$$ \pi_j = \sum_{k\ne j} \pi_k(p_{ij} + p_{ik}p_{ij} + p_{ik}^2p_{ij} + \dots), $$
with the normalized solution,
$$ \pi_j = \frac{p_{ij}(1-p_{ij})} {\sum_k p_{ik}(1-p_{ik})}. $$

The fraction of sequences that starts with $y_{ij}$ and counts $a=\phi_n-1$ arrivals is $\pi_j P_j(\phi_n-1)$ for $n\ge 2$. In such a sequence, the number of inputs for which a $\CLASS(\cdot)$  inference is performed (including the first) is $n-1$. We deduce the overall proportion of inputs that perform a $\CLASS(\cdot)$  inference, 
$$r_i=\frac{\sum_j \sum_n (n-1) \pi_j P_j(\phi_n-1) }{ \sum_j \sum_n (\phi_n-1) \pi_j P_j(\phi_n-1)},$$ 
provided the series converge. It is easy to see that the numerator always converges and the denominator converges when $p_{ij} < 1/\beta$ for all $j$. If $\max_j\{p_{ij}\} \ge 1/\beta$ on the other hand, the denominator is dominated by terms proportional to $(\beta \max_j\{p_{ij}\})^n$ and goes to infinity yielding $r_i=0$. 

The number of inputs in the sequence that are not verified is $(\phi_n-n)$ and each of these brings an error with probability $(1-p_{ij})$. The overall fraction of errors is thus 
$$e_i=\frac{\sum_j \sum_n  (\phi_n-n) (1-p_{ij})\pi_j P_j(\phi_n-1)}{\sum_j \sum_n (\phi_n-1) \pi_j P_j(\phi_n-1)},$$ 
provided the series converge, i.e., the case when $p_{ij}$$<$$1/\beta$ for all $j$. If $\max_j\{p_{ij}\}$$\ge$$1/\beta$, numerator and denominator are dominated by terms proportional to $(\beta \max_j\{p_{ij}\})^n$ whose constant ratio is $1-\max_j\{p_{ij}\}$ yielding expression \eqref{eq:errors}.
\end{proof}

For the $K$ most popular approximate keys, the fraction of inputs requiring $\CLASS(\cdot)$  inference, namely the \emph{refresh rate}, is,
\begin{equation}
R^{(\text{ideal})} = \sum_{1\le i \le K} q_i r_i.
\label{eq:Lideal}
\end{equation}
The overall fraction of inputs requiring inference is then $R^{(\text{ideal})} + (1-H^{(\text{ideal})})$ where $H^{(\text{ideal})}$ is given by \eqref{eq:Hideal}. 
Hence, when the auto-refresh algorithm is used, the error rate for cached items is,
 \begin{equation}
E^{(\text{ideal})} = \sum_{1\le i \le K} q_i e_i,
\label{eq:Eideal}
\end{equation}
and this is the overall fraction of errors due to approximate-key caching since non-cached approximate keys are correctly classified.
\smallskip



\subsection{Impact of error control}\label{sec:model:summary}

Proposition \ref{prop:analysis} illustrates the desirable behavior of auto-refresh with error control: approximate keys with a dominant class yield few errors and rarely require verification while approximate keys with equally likely classes are, on the contrary, frequently verified. 


\subsubsection{Dominant class}
Formulas \eqref{eq:lookups} and \eqref{eq:errors} show that, if an approximate key $x'_i$ has a dominant class such that $\max_j\{p_{ij}\} > 1/\beta$, the fraction of inputs requiring $\CLASS(\cdot)$ inference is asymptotically zero while the error rate is bounded. We have, 
\begin{equation}
r_i = 0, \quad e_i \le 1-\frac{1}{\beta}.
\end{equation}
This shows how the choice of back-off rate in the auto-refresh algorithm trades off accuracy for classification throughput: the smaller $\beta$, the smaller the error, while the number of slower inferences performed via a \CLASS($\cdot$) call is greater. 


\subsubsection{No dominant class}
The worst scenario for an approximate key happens when there are multiple possible classes with equal probability, i.e.,  $p_{ij}=1/m_i$. For the particular case $\beta=2$, the proposition gives,
\begin{equation}
r_i=\frac{m_i-2}{m_i-1}, \quad e_i = \frac{1}{m_i}
\end{equation}
Clearly, when $m_i$ is large, auto-refresh hardly reduces the inference workload (asymptotically, $r_i \rightarrow 1$) and it would be better not to cache this approximate-key at all. On the other hand, the auto-refresh algorithm is still able to maintain a low error rate (asymptotically null, $e_i \rightarrow  0$). Such approximate keys are frequently verified and yield small errors 
which is indeed the desired behavior.

\section{Trace-driven evaluation\label{sec:evaluation}}

We have evaluated approximate-key caching using trace data relating to traffic classification. The presented results are for ideal caching 
for which we obtained closed formulas of easy interpretation.
We have verified by simulation on the same data that LRU behaves similarly but we do not report results here due to lack of space. 
We first introduce details of the use case (Sec. \ref{sec:usecase}) and describe the properties of the dataset (Sec. \ref{sec:dataset}).
We then dig into the auto-refresh performance
(Sec. \ref{sec:performance}) and finally compare approximate-key caching with state-of-the-art similarity caching (Sec. \ref{sec:comparison}).


\begin{figure*}[!t]
\begin{center}
    \subfigure[]{
        \includegraphics[width=0.8\columnwidth]{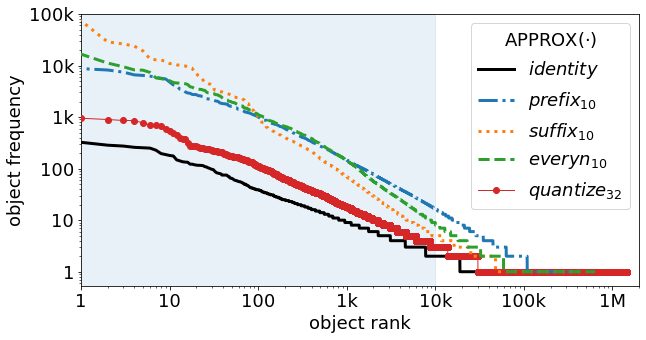}
    }
    \hfill
    \subfigure[]{
        \includegraphics[width=0.55\columnwidth]{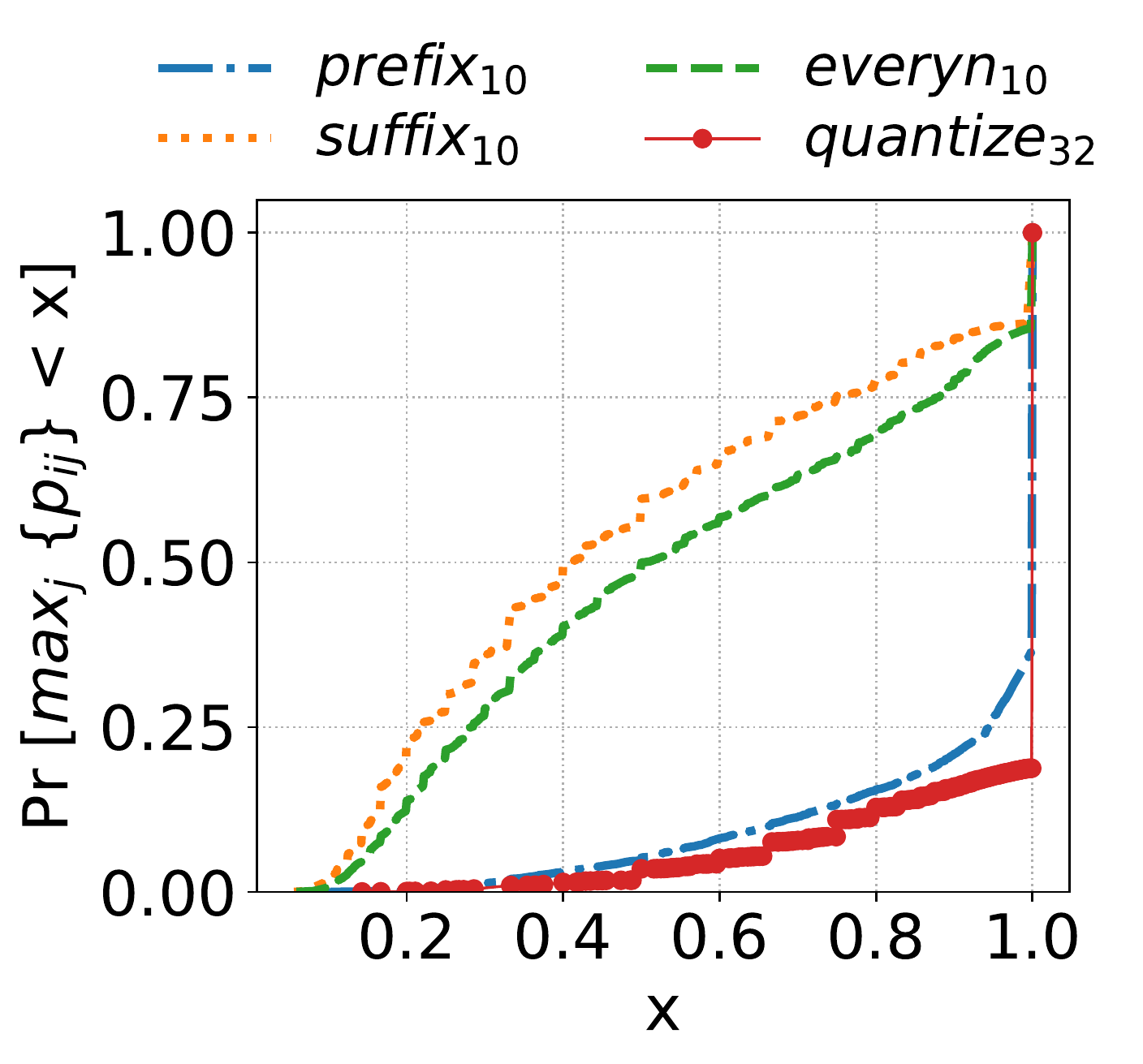}
    }
    \hfill
    \subfigure[]{
        \includegraphics[width=0.55\columnwidth]{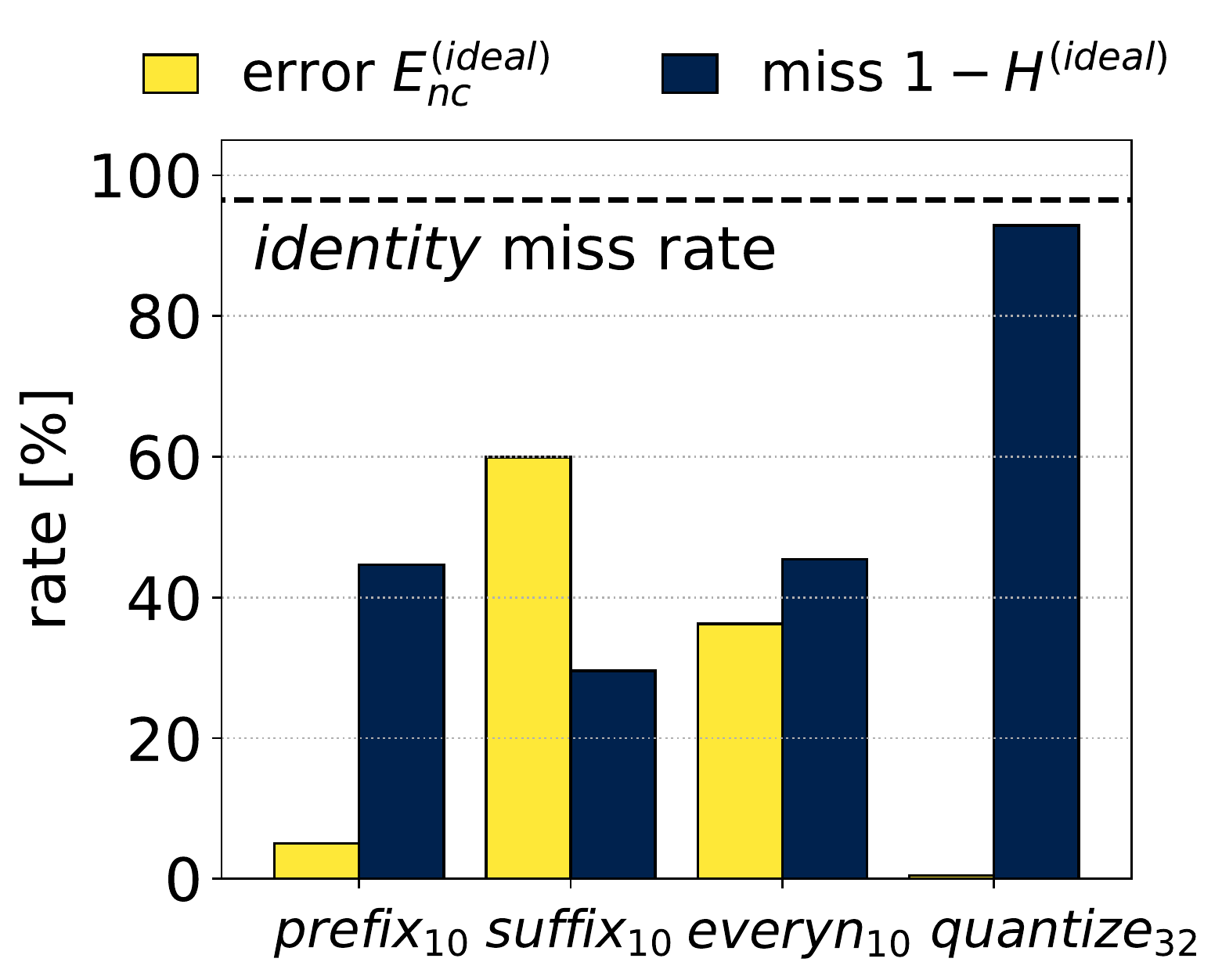}
    }
\caption{Dataset properties: impact of \APPROX($\cdot$) function on object popularity skew (a); dominant label prevalence $\max_j \{p_{ij}\}$ for top 10,000 objects (b); miss rate and error rate without auto-refresh for cache capacity $K=10,000$ (c). 
}
\label{fig:dataset}
\end{center}
\end{figure*}

\subsection{Traffic classification}\label{sec:usecase}
Traffic classification is the act of labeling a network flow with the application that generated it and is a well-known problem~\cite{nguyen08comst,boutaba18comprehensive,aceto-tma18}. 
In particular, \cite{aceto-tma18} compared several DL traffic classifiers with models having different architectures and sizes ranging from 300K to over 6M weights. While not explicitly reported, it is well known that a DL model inference is much slower (hundreds of milliseconds) than a DRAM cache lookup (sub-microsecond). Common off-the-shelf switches and other data-plane programmable hardware can only run very simple ML models~\cite{xiong19hotnets}. Traffic classification is therefore an interesting use case to test the applicability of approximate-key caching as it might help running DL models where there are limited processing resources.

DL traffic classifiers commonly use packet time series as input~\cite{aceto-tma18}, i.e., input $x$ 
is a vector of features (such as size in bytes and direction) of the first $N$ packets of a bi-directional flow.
Results discussed in~\cite{aceto-tma18} only pertain to the accuracy of the DL models. In this paper, we are instead interested in assessing the additional classification errors due to approximate-key caching. We, therefore, use a perfect classification oracle for the $\CLASS(\cdot$) function. 

\subsection{Dataset analysis}\label{sec:dataset}
 
 
We used a private, large-scale dataset,\footnote{We are investigating the possibility to release the anonymized dataset.} comprising over $1$M flows generated by over 76,000 devices. Each flow is represented as a time series of the first 100 packets size (in bytes) and direction (positive or negative) which is coupled with a label (generated from a DPI engine) specifying one among 200 application classes. We split the data into TCP and UDP traffic portions, each of which can be handled by a dedicated 1d Convolutional Neural Network (CNN)  model, with over 90\% accuracy\cite{doubleblind}. 
As TCP and UDP signatures are radically different, for simplicity we report results only for TCP traffic portion (which corresponds to 80\% of bytes and about half of the flows) although we point out results are similar for UDP.

To verify that our proposal works irrespective of the selected $\APPROX(\cdot$) function, we use a subset of the functions introduced earlier in Sec. \ref{sec:AKC}, namely, \emph{prefix}$_n$, \emph{suffix}$_n$, \emph{everyn}$_n$, \emph{quantize}$_n$, for selected values of $n$. 
The approximations significantly reduce the size of the input space (smaller vector dimension or smaller elements) and increase the popularity distribution skew while, of course, introducing undesirable errors, as evaluated below.  

\subsubsection{Popularity skew and dominant classes}
 Fig.~\ref{fig:dataset}-(a) shows the impact of selected $\APPROX(\cdot)$ functions on the popularity skew in the transformed space $\mathcal{X}'$ compared to the original space $\mathcal{X}$ (denoted ``identity'' function).  
It is evident, especially for  \emph{prefix$_{10}$}, \emph{suffix$_{10}$}, and \emph{everyn$_{10}$}, that the frequency of popular inputs increases significantly. This clearly improves the potential hit rate compared to that of the original trace.   


Fig.~\ref{fig:dataset}-(b) shows the probability distribution of $\max_j \{p_{ij}\}$ for the considered $\APPROX(\cdot)$ functions for the top 10,000 objects (shaded area in Fig.~\ref{fig:dataset}-(a)). We observe, especially for \emph{quantize$_{32}$} and \emph{prefix$_{10}$}, that there is a high proportion of approximate keys having a dominant label. As discussed in Sec. \ref{sec:model}, this suggests the resulting error will be small and that \emph{quantize$_{32}$} and \emph{prefix$_{10}$} are better $\APPROX(\cdot$) functions. 


\subsubsection{Hit rates and error rates}
We first assess the hit rate and error rate induced by $\APPROX(\cdot)$ alone, i.e., approximate-key caching without error correction.
For a cache size $K=10,000$ elements, Fig.~\ref{fig:dataset}-(c) shows that the increased popularity skew significantly reduces the fraction of inputs requiring $\CLASS(\cdot)$ inference, as captured by the miss rate $1-H^{\text{(ideal)}}$. The miss rate decreases from over 95\% for exact caching and \emph{quantize$_{32}$} to about 30\% for \emph{suffix$_{10}$} and to less than 50\% for \emph{prefix$_{10}$} and \emph{everyn$_{10}$}. Approximate-key caching can therefore potentially halve the number of $\CLASS(\cdot)$ inferences.

At the same time, approximate-key caching introduces errors whose rate depends on the $\APPROX(\cdot)$ function. Fig.~\ref{fig:dataset}-(c) also shows that the error rate ranges between 5\% (for \emph{prefix$_{10}$}) and 60\% (for \emph{suffix$_{10}$}). As per the previous analysis, we expect auto-refresh to compensate for such errors by more frequently verifying cached approximate keys mapping to multiple classes.  
Conversely, when such errors are small (e.g., for \emph{prefix$_{10}$} or \emph{quantize$_{32}$}), a few verification cycles will suffice to further reduce the error without noticeable effect on the hit rate.
If errors are frequent for the considered  $\APPROX(\cdot)$, this may offset the hit rate benefits as a verification essentially boils down to running an actual $\CLASS(\cdot)$ inference. In the worst case of a badly designed $\APPROX(\cdot)$ function, we expect  auto-refresh to correct cache mismatches for the erroneous prefix-label mappings in a seamless manner.



\begin{figure}
    \centering
    \includegraphics[width=0.87\columnwidth]{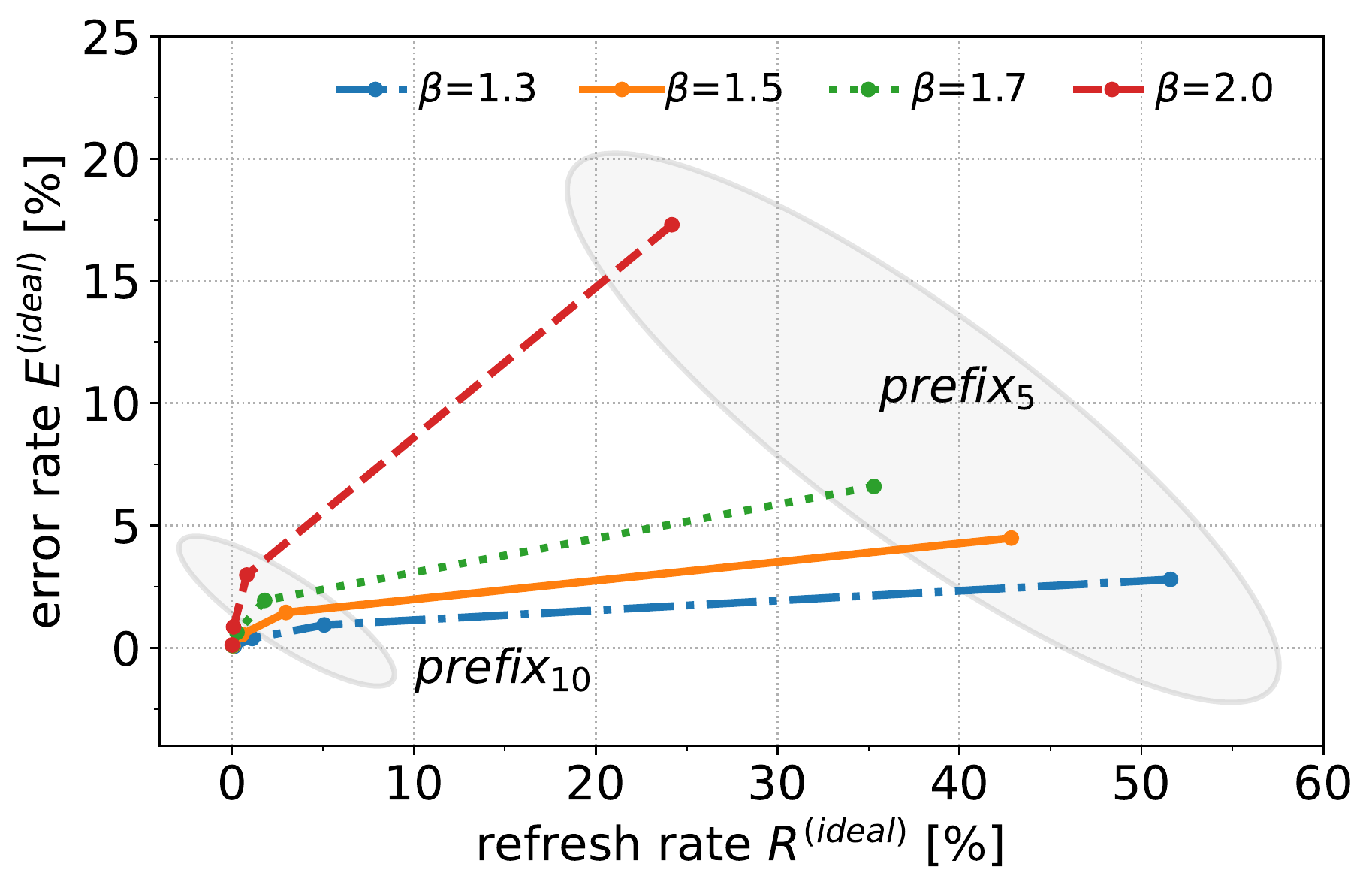}
    \caption{Auto-refresh performance: error-vs-refresh rates trade-off induced by different back-off values $\beta$ for the $\APPROX(\cdot)$=\emph{prefix$_n$} function family.
    \label{fig:autorefresh-backoff}}
 \end{figure}

\subsection{Auto-refresh performance}\label{sec:performance}
We explore auto-refresh performance ($i$) by varying the exponential back-off base $\beta$ and ($ii$) by varying the $\APPROX(\cdot)$ function. The cache size is set here to $K=10,000$ but other capacities have similar qualitative results. 

\subsubsection{Back-off parameter}
We focus here on the \emph{prefix$_n$} approximation function that was shown above to have a favorable hit rate vs error rate trade-off.  This trade-off in fact depends on the value of $n$: the popularity skew increases as $n$ gets smaller, bringing a higher hit rate, but this comes at the cost of a higher error rate since fewer approximate keys have a dominant class. When $n$ is large, on the other hand, the performance of \emph{prefix$_n$} tends to that of the \emph{identity} function.

Fig.~\ref{fig:autorefresh-backoff} depicts, for different \emph{prefix$_n$} function sizes $n$$\in$$\{5,10,20,50\}$, the refresh rate $R^{\text{(ideal)}}$ vs error rate $E^{\text{(ideal)}}$ trade-off realized with different back-off values $\beta$.
We observe that, while performance varies as a function of $n$ as expected (i.e., shorter prefixes lead to higher errors) the impact of the back-off is consistent (i.e., lines do not cross).
In the rather extreme case of $n=5$, the error rate without error correction amounts to 45\%. Notice that when $\beta=2$, the auto-refresh limits the error to about 17\% with a refresh rate of 25\% (i.e., one every four hits is verified); on the other hand, by setting $\beta=1.3$ it is possible to further halve the error to 8\%, but at the cost of an increased refresh rate of about 55\%.
The same qualitative trade-off holds for other settings (e.g., $n=10$ in the picture) allowing one to tune the auto-refresh performance through the choice of $\beta$ depending on the use case (e.g., whether errors are tolerable and $\CLASS(\cdot$) is the dominant cost or, on the contrary, the error needs to be bounded). 

While in this work we statically set $\beta$, we argue that $\beta$ could be tuned dynamically  to maintain the $\CLASS(\cdot$) inference rate below a target value. A similar load control objective is realized by ``model switching''~\cite{modelswithc-HOTCLOUD20} where, depending on the load, models with different complexity (thus different accuracy) are interchanged, i.e., under high demand, use models with lower accuracy but faster inference, and vice versa. We believe the present proposal to use approximate-key caching and dynamically vary $\beta$ will realize this objective more simply. We leave the evaluation of such mechanisms to future work.


\begin{figure}[!t]
    \centering
    \includegraphics[width=0.95\columnwidth]{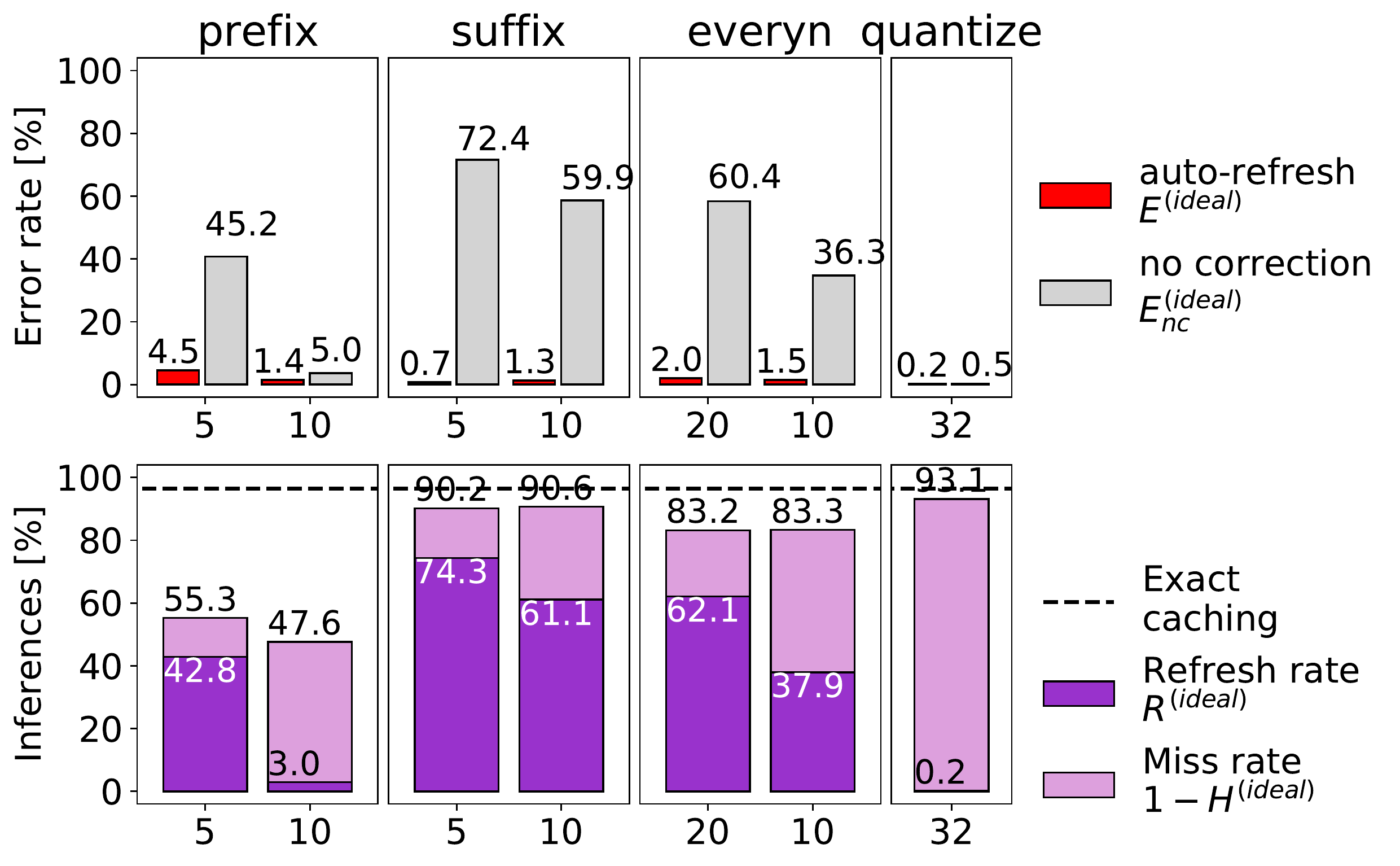}
    \caption{Auto-refresh performance:  overall costs-vs-benefits for all $\APPROX(\cdot)$ functions when $\beta=1.5$.}\label{fig:autorefresh-approx}
 \end{figure}

\subsubsection{Approximation functions}
We now fix the back-off rate to $\beta=1.5$ and assess the auto-refresh performance for a wide range of $\APPROX(\cdot)$ functions.
Fig.~\ref{fig:autorefresh-approx}-(top) compares the error  
with ($E^{\text{(ideal)}}$, red bars) and without ($E_{nc}^{\text{(ideal)}}$, grey bars) auto-refresh. 
Results show that the proposed mechanism is able to successfully correct even very large errors (in excess of 70\%), leading to a remarkably low error rate (generally a few percentage points). 

Clearly, error reduction comes at the cost of a higher refresh rate $R^{(\text{ideal})}$. Fig.~\ref{fig:autorefresh-approx}-(bottom) reports the overall $\CLASS(\cdot$) inference rate, by stacking the refresh rate $R^{(\text{ideal})}$ (verification of cache values, dark color) with the miss rate $1-H^{(\text{ideal})}$ (inferences for objects not in the top 10,000, light color). We observe that, with a few exceptions, the majority of the inferences are due to auto-refresh.
Notice that \emph{prefix$_n$} yields the best results, particularly when $n=10$  (half the miss rate with respect to exact caching, with only 3\% verification required to reach 1.4\% error), though the error correction mechanism yields satisfactory results even for admittedly sub-optimal $n=5$  settings ($\simeq$ half the miss rate, 42\% of verification to reach 4.5\% of error). This reinforces the generality of the approach and  alleviates the burden of designing and tuning the $\APPROX(\cdot)$ functions, as basic domain knowledge on the use case at hand suffices to obtain good results.

\subsection{Comparison to Similarity caching}\label{sec:comparison}

We here compare approximate-key caching with similarity caching focusing  on two aspects: ($i$) the computational complexity of the lookup operation and ($ii$) the accuracy of the returned results. We further include complexity results for exact caching as a reference.

\begin{figure}[!t]
    \centering
    \includegraphics[width=0.92\columnwidth]{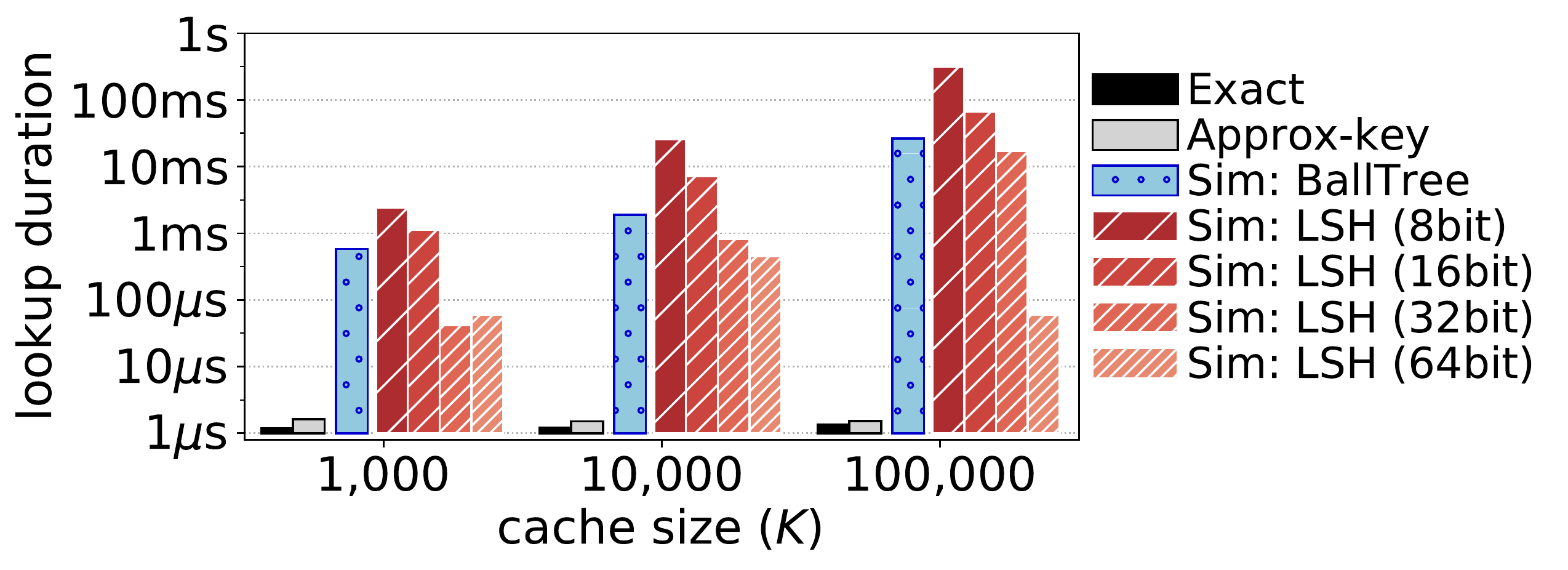}
    \includegraphics[width=0.92\columnwidth]{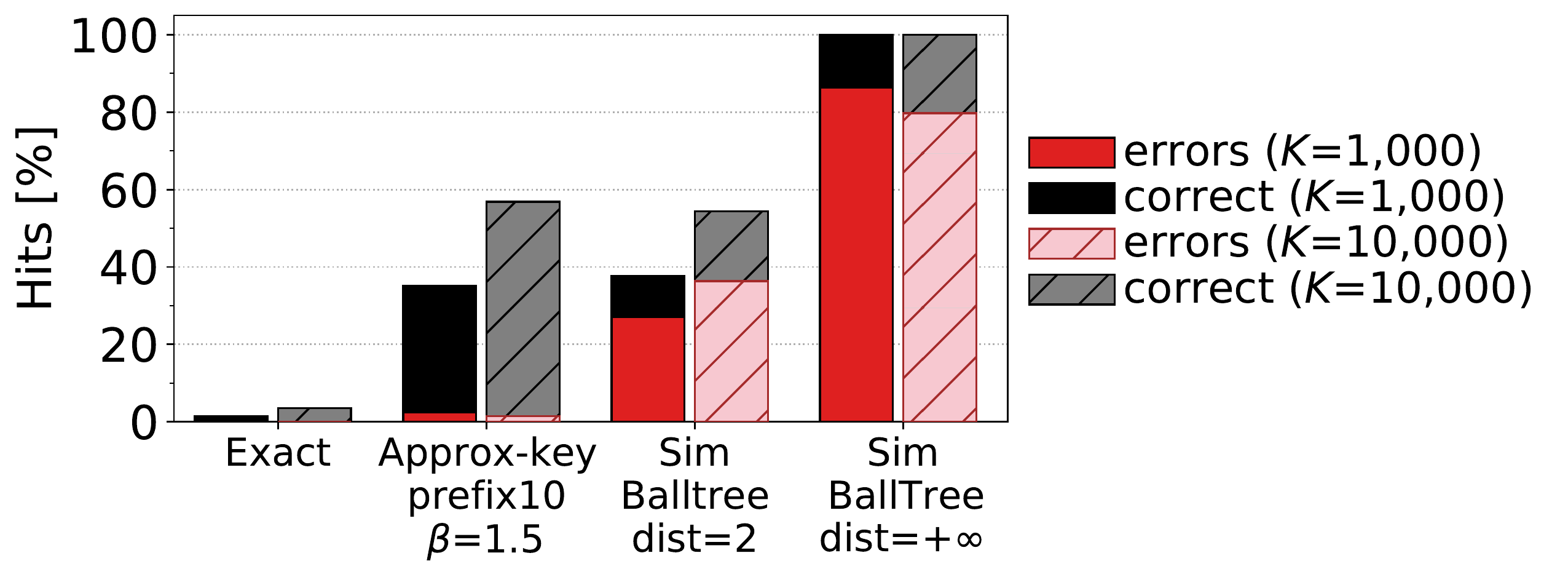}
    \caption{Comparison to similarity caching: lookup duration 
    for different implementations and settings (top) and 
    hit rate breakdown between errors and correct answers (bottom).}
    \label{fig:similarity}
\end{figure}

\subsubsection{Implementation details}

All evaluations are performed using \texttt{Python 3.7}. Specifically, for exact caching and approximate-key caching we resort to the native Python \emph{dictionary} (i.e., key-value paired hash tables) while implementing $\APPROX(\cdot)$ functions in Python.
For similarity caching, we consider two alternative state-of-the-art kNN methods~\cite{knn-book}: \emph{BallTree} as offered by the \texttt{scikit-learn} v0.23.2 library~\cite{sklearn-balltree}, and LSH as implemented by the \texttt{lshashpy3} v0.0.8 library~\cite{lshashpy3}. In a nutshell, a BallTree is a binary tree partitioning the search space based on pair-wise distances~\cite{pynndescent}; LSH instead partitions the search space using hash functions based on randomized Gaussian projections~\cite{slaney08lsh}. We note that the authors of these libraries have optimized their implementations by offloading complex operations to C libraries.\footnote{BallTrees use Cython~\cite{cython}, while \texttt{lshashpy3} uses \texttt{numpy}~\cite{numpy}.} We used the default parameters and Euclidean distance to train a BallTree, while we run LSH with a single hash table (more hash tables just degraded the performance). In both cases, we search for the 10 closest neighbors and apply majority voting to select the output label. We also tested alternative popular libraries like  \texttt{PyNNDescent}~\cite{pynndescent} and observed no performance differences. 


\subsubsection{Lookup duration}
Fig.\ref{fig:similarity}-(top)  depicts the average lookup duration for various cache sizes $K$$\in$$\{10^3,10^4,10^5\}$ and caching paradigms. We train BallTrees and LSH caches using the equivalent top-$K$ objects (we tested also with a random selection with no performance difference). We used \emph{prefix}$_{10}$ as $\APPROX(\cdot)$ function.
The evaluation was performed on a Linux server equipped with Intel Xeon Platinum 8164 CPUs @ 2.00GHz. While the precise numerical results are relevant only for this architecture, the \emph{relative} performance of the different paradigms remains relevant for alternative implementations.  

First, as expected, the cost of the approximate-key caching lookup is only marginally higher than exact caching lookup (due to the $\APPROX(\cdot)$ computation) with both remaining on the order of a microsecond for all explored cache sizes. Second, similarity caching lookup duration is between 2 to 5 orders of magnitude larger and the penalty with respect to approximate-key caching grows noticeably with the cache size $K$. Third, while LSH may provide a faster lookup than BallTree, its performance is sensitive to the settings used (number of bits for the hash size, and number of hash tables). Moreover, the lookup duration remains significantly higher for LSH than for approximate-key caching. The literature on this aspect is also divided. In fact, while LSH is often cited as the solution to speed up similarity caching~\cite{chierichetti09pods,garetto20infocom}, others have stated
``\emph{As we know, LSH
does not perform well on the ANN [approximate nearest neighbor] problem compared with the fancy optimization-based hashing methods. This is mainly
due to the random nature of its hash functions that are too
weak to capture the complex distribution of input features}''~\cite{ding18transaction}, and we remark that \texttt{scikit-learn} also removed LSH due to  speed concerns~\cite{sklearn-lsh-removed}.

Finally, note that a lookup duration of $\simeq$100ms starts being comparable with the $\CLASS(\cdot)$ inference duration that similarity caching is  meant to reduce, putting in question the usefulness of similarity caching for this use case.

\subsubsection{Accuracy of approximated answers} 
Fig.\ref{fig:similarity}-(bottom) shows how approximate-key caching enables better error control with respect to similarity caching. In particular, we break down hits between erroneous ones (red/pink bars) and correct ones (black/grey bars stacked on top).
Although we can tune a BallTree to have a hit rate comparable to that of approximate-key caching  (distance threshold $\epsilon$ equal to 2), more than 65\% of the hits are errors.  In contrast, approximate-key caching generates less than 2\% errors. Reducing the distance leads to fewer hits overall but does not reduce the proportion of errors. 
LSH yields similar performance results that are not reported in the figure for the sake of brevity.
In other words, similarity caching is prone to errors since classes cannot be easily separated in the input space. An accurate classifier indeed requires  a complex DL model.


\section{Related work}\label{sec:related}
Related work on approximate caching can be found in the areas of Deep Learning\cite{lcnn-CVPR17,deepcache-mobicom18,freezing-hotcloud19,clipper-NSDI17,laser-WSDM14,velox-CIDR15}, content retrieval \cite{virage-96,falchi-LSDS08,chierichetti09pods,isax-icdm10,imcom17,lucene-10,pandey-WWW09,li-WWW18} and networking\cite{softcachehits-jsac18,sensornetwork-12}, with either a system or a theoretical flavor.

\subsubsection{Application domain} 
Similarity caching has been applied to a variety of domains including image search~\cite{virage-96,falchi-LSDS08,isax-icdm10,imcom17}, text search~\cite{lucene-10}, ads recommendation~\cite{pandey-WWW09,li-WWW18}, multimedia~\cite{softcachehits-jsac18}, and sensor networks~\cite{sensornetwork-12}. Notice that none of these scenarios is related to classification tasks.
A number of proposals~\cite{lcnn-CVPR17,deepcache-mobicom18,freezing-hotcloud19} aim to accelerate DL inference by \emph{caching partial results} at intermediate feed-forward layers and are thus orthogonal to our work as they require acting on the inner mechanics of DL models. Conversely, several DL inference systems~\cite{clipper-NSDI17,laser-WSDM14,velox-CIDR15} simply use \emph{exact caching} and might greatly benefit from our proposal.  

\subsubsection{Similarity caching} 
Similarity caching uses a kNN search. From a practical viewpoint, while multiple surveys~\cite{simsearch-survey06,hydra-vldb19} overview the abundant literature on kNN algorithms and data structures, their computational cost remains
high~\cite{fastscan-vldb15,lshgpu-tobd21}
requiring hardware-specific acceleration or multi-core processing~\cite{sundaram13vldb}.
From a modeling viewpoint, 
the theory behind similarity caching is still in its infancy~\cite{chierichetti09pods,pavlos18jsac,garetto20infocom,sabnis21infocom}. In~\cite{chierichetti09pods}, the authors studied the approximation space though the similarity function they use is a heuristic. Authors of \cite{pavlos18jsac} instead frame similarity caching as an optimization problem but fail to provide a dynamic policy. More recently,
\cite{garetto20infocom} investigated the \emph{existence} of an optimal dynamic policy.
However, none of these works directly apply to classification. 
 
\subsubsection{Error control} 
To the best of our knowledge, no work exists that \emph{explicitly} tackles the issue of controlling the error in the approximation.
A few works~\cite{pandey-WWW09, garetto20infocom,li-WWW18,sabnis21infocom} discuss the use of cost functions to trigger cache entry eviction according to the similarity search outcome -- with somewhat complex means, e.g., a 
gradient boost regression tree to model the cost function\cite{li-WWW18} or with
gradient descent to discover the best set of cache entries~\cite{sabnis21infocom}. However, the error is not formally analyzed, nor explicitly controlled -- which is one of the major contributions of this work.

\section{Conclusion}\label{sec:conclusion}

This paper introduces approximate-key caching, a new caching paradigm for classification tasks that
retains the simplicity of exact caching, while increasing the cache hit rate by significantly reducing the size and skew of the input space. 
Additionally, approximate-key caching incorporates a novel auto-refresh algorithm that controls the impact of errors by explicitly verifying key-value mappings for selected input queries. The algorithm has been analytically modeled and thoroughly evaluated using trace data relating to a traffic classification use case. Overall, our work shows that approximate-key caching is robust (as the auto-refresh mechanism can significantly reduce even a very large rate of errors), simple (as cache lookups are orders of magnitude faster than for similarity caching), and effective (as it is easy to define $\APPROX(\cdot)$ functions that considerably reduce the number of classification inferences needed).


\bibliographystyle{IEEEtran}
\bibliography{biblio}

\end{document}